\documentclass{article}
\usepackage[utf8]{inputenc}
\usepackage[english]{babel}

\usepackage{blindtext}
\usepackage{amssymb}
\usepackage{amsmath}
\usepackage{comment}
\usepackage{mleftright}
\usepackage{xcolor}

\usepackage{amsthm}
\usepackage{amsmath}

\newtheorem{theorem}{Theorem}[section]

\newtheorem{lemma}[theorem]{Lemma}
\newtheorem{proposition}[theorem]{Proposition}

\newtheorem{remark}{Remark}

\theoremstyle{definition}
\newtheorem{definition}{Definition}[section]
\newtheorem{example}{Example}[section]

\allowdisplaybreaks 

\newcommand{\n}{n}

\newcommand{\ee}{\varepsilon}
\newcommand{\nn}{\n \otimes \n}
\newcommand{\id}{\textup{I}}
\newcommand{\curl}{\mathrm{curl}\,}
\newcommand{\dive}{\mathrm{div}\,}
\newcommand{\tr}{\mathrm{tr}\,}
\newcommand{\R}{\mathbb{R}}

\title{A Novel Landau-de Gennes Model with Quartic Elastic Terms}
\date{\today}

\usepackage[capitalise]{cleveref}

\usepackage{authblk}
\numberwithin{equation}{section}
\begin{document}
\renewcommand\Authfont{\small}
\renewcommand\Affilfont{\itshape\footnotesize}

\author[1]{Dmitry Golovaty\footnote{dmitry@uakron.edu}}
\author[3]{Michael Novack\footnote{mrnovack@indiana.edu}}
\author[3]{Peter Sternberg\footnote{sternber@indiana.edu}}
\affil[1]{Department of Mathematics, University of Akron, Akron, OH 44325}
\affil[2,3]{Department of Mathematics, Indiana University, Bloomington, IN 47405}

\maketitle

\noindent {\bf Abstract:} Within the framework of the generalized Landau-de Gennes theory, we identify a $Q$-tensor-based energy that reduces to the four-constant Oseen-Frank energy when it is considered over orientable uniaxial nematic states. Although the commonly considered version of the Landau-de Gennes theory has an elastic contribution that is at most cubic in components of the $Q$-tensor and their derivatives, the alternative offered here is quartic in these variables. One clear advantage of our approach over the cubic theory is that the associated minimization problem is well-posed for a significantly wider choice of elastic constants. In particular, this quartic energy can be used to model nematic-to-isotropic phase transitions for highly disparate elastic constants.

In addition to proving well-posedness of the proposed version of the Landau-de Gennes theory, we establish a rigorous connection between this theory and its Oseen-Frank counterpart via a $\Gamma$-convergence argument in the limit of vanishing nematic correlation length. We also prove strong convergence of the associated minimizers.

\section{Introduction}

Two well-regarded and heavily researched mathematical models for nematic liquid crystals are the director-based Oseen-Frank model \cite{Os33,Zo33,Fr58,virg} and the $Q$-tensor-based Landau-de Gennes theory \cite{de71,dePr95,virg,newt}. In this article we present a version of Landau-de Gennes which enables us to make a rigorous and relatively simple asymptotic connection between the two theories in the limit of Landau-de Gennes with vanishing non-dimensional nematic correlation length. In \cite{MaZa10}, the authors carry out such a program in the so-called `equal-constants' setting where the elastic energy of both models is given by the Dirichlet integral. Here we identify an elastic energy density involving certain terms that are quartic in the $Q$-tensor and its derivatives such that the corresponding Landau-de Gennes energy approaches, in the sense of $\Gamma$-convergence, the full Oseen-Frank energy.
 
 We recall now the form of the Oseen-Frank energy,
 \begin{align}\notag
F_{OF}(n):=&\int_\Omega\left(\frac{K_1}{2}(\dive n)^2 + \frac{K_2}{2}((\curl n) \cdot n)^2 +\frac{K_3}{2}|(\curl n )\times n|^2\right. \\ \label{FOF1}
&\left. \quad\quad +\frac{K_2+K_4}{2}(\textup{tr}\, (\nabla n)^2 - (\dive n)^2)\right)\,dx,
\end{align}
where $K_1,K_2,K_3$ and $K_4$ are material constants,  $\Omega\subset\R^3$ represents the sample domain and the director $n$ maps 
$\Omega$ to $\mathbb{S}^2$.

We also recall the standard Landau-de Gennes model defined for $Q:\Omega\to \mathcal{S}$, where 
 \[\mathcal{S}:=\left\{Q\in M^{3\times3}:Q^T=Q,\,\tr{Q}=0\right\}.\] 
It is given by
\begin{equation}
\label{elastic}
F_{LdG}(Q):=\int_\Omega\left(\frac{L_1}{2}{|\nabla Q|}^2+\frac{L_2}{2}Q_{ij,j}Q_{ik,k}+\frac{L_3}{2}Q_{ik,j}Q_{ij,k} +W(Q)\right)\,dx,
\end{equation}
{where $Q_{ij,k}=\partial Q_{ij}/\partial x_k$ and repeated indices are summed from one to three.}
Here the bulk Landau-de Gennes energy density is
\begin{equation}
\label{LdG}
W(Q):=a\,\mathrm{tr}\left(Q^2\right)-\frac{2b}{3}\,\mathrm{tr}\left(Q^3\right)+\frac{c}{2}\left(\mathrm{tr}\left(Q^2\right)\right)^2,
\end{equation}
cf. \cite{gartland}. The coefficient $a$ is temperature-dependent and in particular is negative for sufficiently low temperatures. Throughout this article we will assume that we are in a temperature regime where $a<\frac{b^2}{27c}$, an inequality implying that $Q$-tensors in the minimal set $\mathcal N$ of $W$ are in a uniaxial nematic state describable in terms of a director as follows
\begin{equation}
\mathcal N:=\left\{s_0 \left( \n \otimes \n- \frac{1}{3}\id\right): \n \in \mathbb{S}^2 \right\}.\label{minQ}
\end{equation}
Here $s_0$ is given explicitly in terms of the coefficients $a$, $b$, and $c$, and by subtracting an appropriate constant from $W$ one can take $W$ to vanish along this minimal set. We will ignore this constant and simply assume without loss of generality that $W$ vanishes along this minimal set of states. 

The effort to find a connection between $F_{OF}$ and a corresponding $Q$-tensor-based elastic energy has a long history, going back at least to \cite{BeMe84}, and includes the contributions of \cite{Di95,LoTr,Na97}.
However, these studies are premised on the observation that one can obtain $F_{OF}$ from a $Q$-tensor model {for certain values of the elastic constants by adding to the standard Landau-de Gennes energy \eqref{elastic} an additional elastic term that is cubic in $Q$ and its derivatives, namely
\[
\int_{\Omega}Q_{lk}\partial_kQ_{ij}\partial_lQ_{ij}\,dx.
\]
As pointed out, for instance, in \cite{ball2017liquid}, pg. 21, this choice represents one of the six possible linearly independent cubic terms that are quadratic in $\nabla Q$ and respect the necessary symmetries.}
From the standpoint of energy minimization, unfortunately, such a version of Landau-de Gennes becomes problematic, since the inclusion of the cubic term leads to an energy which is unbounded from below, \cite{BaMa10}. Indeed, quoting \cite{LoTr}, ``In the presence of biaxial fluctuations the general third order theory in $Q_{\alpha\beta}$ becomes unstable and thus is thermodynamically incorrect. One has to include higher order terms (or neglect third-order ones) to preserve stability of the free energy." Alternatively, one can impose a constraint through a choice of bulk potential that penalizes large $Q$ and prevents the cubic elastic term(s) from overtaking the quadratic ones, {\cite{BaMa10,bp,fatkullinslastikov,kkls}.}

Not surprisingly, this deficiency then also leads to instabilities in attempts to capture dynamics through the corresponding gradient flow. Along these lines we mention the work of  \cite{IyXuZa15}
where the authors overcome this impediment to obtain a dynamical well-posedness result under an assumption of sufficiently small initial data, while also showing blow-up for large initial data.
 
 In contrast to these sizable troubles to be overcome when taking a cubic elastic energy density for Landau-de Gennes, we propose a version of Landau-de Gennes, in the spirit of the quotation from \cite{LoTr} above, involving a quartic elastic energy density that presents none of these technical difficulties. As usual, it is defined over the class $\mathcal{S}$. In a prototypical form, the energy is given by
\begin{align}\notag
\mathcal{F}_{LdG}(Q) := \int_\Omega &\left(\frac{L_1}{2} \left|\left(\frac{s_0}{3}\id + Q\right)\dive Q\right|^2+ \frac{L_2}{2} \left|\left(\frac{s_0}{3}\id + Q\right) \curl Q\right|^2\right.\\*
&\left. +\frac{L_3}{2}\left|\left(\frac{2s_0}{3}\id - Q\right) \dive Q \right|^2 +\frac{L_4}{2}\left|\left(\frac{2s_0}{3}\id-Q\right)\curl Q\right|^2\right.\nonumber\\*
&\qquad + W(Q)\bigg)\,dx.\label{ourfldg}
\end{align}
Here the elastic constants $\{L_i\}$ are taken to be positive and we assume $W$ is still given by \eqref{LdG}. We point out that our model is still {\em quadratic} in its dependence on the gradient of $Q$.

{ Before explaining how we arrive at $\mathcal{F}_{LdG}$ we want to be clear on our motivation for seeking such a version of Landau-de Gennes energy. Our criteria were:\\
$\bullet$ For a reasonable range of elastic constants, one should be able to recover the four-term Oseen-Frank energy among uniaxial $Q$-tensors as in \eqref{minQ}. This range should, in particular, allow for the regime of extreme disparity between the $K_i$'s, since we wish to use the model to explore various types of liquid crystals for which some subset of the three deformations splay, twist and bend is far more favorable energetically than others. For example, we seek a model capable of capturing the formation of tactoids in nematic/isotropic phase transitions. (See, e.g. \cite{GSV, GoNoStVe18, GoKiLaNoSt19}.)\\
$\bullet$ Minimization of $\mathcal{F}_{LdG}$ via the direct method should be achievable, and in particular, the energy should be bounded from below and coercive.\\
$\bullet$ Any elastic terms should respect the necessary symmetries and so be selected from the list to be found, for example, in \cite{LoMoTr87}. Among the quartic choices--that is, quadratic in both $Q$ and $\nabla Q$, there are 13 to work with.\\
$\bullet$ One should strive for as simple a choice as possible that meets the previous three criteria.
}

To arrive at $\mathcal{F}_{LdG}$, we begin by considering $Q\in \mathcal N$ where the corresponding director field $n$ is sufficiently smooth. We show that each elastic term in the Oseen-Frank energy can be realized through projections of $\dive Q$ and $\curl Q$ on $n$ and on the plane perpendicular to $n$. This allows us to rewrite $F_{OF}$ in terms of the divergence and curl of $Q\in\mathcal N$. We subsequently relax the constraint $Q(x)\in\mathcal N$ to allow biaxial states as well by assuming simply that $Q$ takes values in $\mathcal S$ and add $W(Q)$ to the elastic energy density to force energy minimizing configurations to have values close to $\mathcal N$. { While we certainly do not claim that our choice is unique, this version of Landau-de Gennes model leads to a variational problem that is well-posed under minimization and rigorously reduces to $F_{OF}$ in the limit of vanishing non-dimensional nematic correlation length. Indeed, we feel the resulting energy $\mathcal{F}_{LdG}$ given in \eqref{ourfldg} meets all the criteria listed above.}


Though in this article we do not address dynamics, in \cite{GoKiLaNoSt19} we carry out computations in the context of an associated gradient flow in a thin film limit. The computations are performed for a director that lies in the plane of the film and in the regime where splay is heavily penalized, i.e. where $L_1$ in $\mathcal{F}_{LdG}$ is much larger than the other elastic coefficients. In doing so, we also pursue a temperature regime for the Landau-de Gennes potential $W$ where both the nematic and isotropic states are preferable. We find nice agreement with certain experimentally observed morphologies associated with tactoid evolution and defect splitting. Indeed, this highlights another favorable feature of the energy $\mathcal{F}_{LdG}$, namely that it allows one to model nematic/isotropic phase transitions in such a way that in the uniaxial nematic region, the energy agrees with Oseen-Frank.

The plan of the paper is as follows. In Proposition \ref{reduction}, we show that for $Q$ in the uniaxial minimizing set given by \eqref{minQ}, $\mathcal{F}_{LdG}(Q)$ reduces to precisely $F_{OF}(n)$ for a particular set of constants $\{L_i\}$ given in terms of $\{K_i\}$. More precisely, we can assert this equivalence when $Q$ is orientable, in the sense that  $Q\in H^1(\Omega,\mathcal N)$ is representable as $Q=s_0(\nn - \id/3)$ for some `lifting' $n\in H^1(\Omega;\mathbb{S}^2)$, a property that in particular always holds when $\Omega$ is taken to be simply-connected, cf. \cite{BaZa11}.  After a non-dimensionalization, leading to a dimensionless version of $\mathcal{F}_{LdG}$, namely $\mathcal{F}_{\varepsilon}$ given in \eqref{Feps} below,
we argue that 
this form of Landau-de Gennes is coercive over $H^1(\Omega;\mathcal{S})$, and weakly lower-semi-continuous, making it well-suited for minimization via the direct method in the calculus of variations when a Dirichlet (strong anchoring) condition is imposed on $\partial\Omega$, cf. Theorem \ref{existence}. As described at the outset of Section 2.2, the parameter $\varepsilon$ appearing in the non-dimensionalization represents a ratio of the nematic correlation length to a characteristic lengthscale of the domain. 

To make rigorous the asymptotic connection between $\mathcal{F}_{\varepsilon}$ and $F_{OF}$ in the limit of small nematic correlation length, we then establish $\Gamma$-convergence and compactness with respect to weak $H^1$-convergence, cf.  Theorem \ref{tt1} and Proposition \ref{compactness}. From standard $\Gamma$-convergence theory this implies the weak $H^1$-convergence of minimizers of Landau-de Gennes to a minimizer of Oseen-Frank. In our last result, Theorem \ref{strongcon}, we upgrade this convergence to strong $H^1$-convergence.

\vskip.1in
\noindent{\bf Acknowledgments.} {\it  DG acknowledges the support from NSF DMS-1729538. MN and PS acknowledge the support from a Simons Collaboration grant 585520.}

\section{Tensor Formulation of the Oseen-Frank\\ Energy}
In this section, we present several calculations which establish equalities between the terms in the Oseen-Frank energy of an $\mathbb{S}^2$-valued vector field $n$ and quartic terms in $\nn$ and $\nabla (\nn)$. These calculations form the basis for our choice of elastic terms for a Landau-de Gennes energy. {Let us first establish the notation used throughout this section and the rest of the manuscript.
\begin{definition}
For a function $f:\mathbb{R}^n \to \mathbb{R}^m$, we define $\nabla f$ to be matrix of partial derivatives $\nabla f = (\partial f_i /\partial x_j) \in \mathbb{R}^{m\times n}$.\end{definition}
When $f$ is scalar-valued, we will for convenience sometimes treat $\nabla f$ as a column vector as opposed to a row vector, to aid in calculations.\\

We will often use the notation $A_j$ for the $j$-th row of a matrix $A$.
\begin{definition}
For a smooth, matrix-valued map $A:\mathbb{R}^n \to \mathbb{R}^{n \times n}$, the vector field $\dive A: \mathbb{R}^n \to \mathbb{R}^n$ is given by
$$
\dive A = \sum_{j=1}^n (\dive A_j) e_j
,$$
so that the $j$-th entry of $\dive A$ is the divergence of the $j$-th row of $A$.
\end{definition}
\begin{definition}
We define the curl of a tensor field $A$ by 
\begin{equation}\label{curldef}
(\curl A) v := \curl (A^Tv) \quad \textup{   for all }v \in \mathbb{R}^3,
\end{equation}
which is equivalent to defining $\curl A$ via
\begin{equation}\notag
\curl A=\ee_{ijk}A_{mj,i}e_k \otimes e_m.
\end{equation}
Hence the $j$-th column of $\curl A$ is the curl $A_j$.
\end{definition}
In order to calculate terms involving the curl of a symmetric tensor, we need the following lemma, the proof of which is immediate from the previous definition.
\begin{lemma}\label{tensorcurl}
\begin{enumerate}
\item For any smooth vector field $m$, we have
$$\curl (m \otimes m) = ((\curl  m_jm)_i) \in \mathbb{R}^{3\times 3} .$$
That is, the $j$-th column of $\curl (m \otimes m)$ is $\curl (m_j m)$.
\item For any tensor field $Q$ taking values in the space of symmetric matrices, if we refer to the $j$-th row of $Q$ as $Q_j$, we have
\begin{equation}\label{qcurl}
\curl Q = ( (\curl Q_j)_i ) \in \mathbb{R}^{3 \times 3}.
\end{equation}
\end{enumerate}
\end{lemma}
We begin our analysis with expressions for the $K_1$, $K_2$ and $K_3$ elastic terms from Oseen-Frank in terms of derivatives of $\nn$.}
\begin{proposition}\label{K1}
Let $\n$ be a smooth vector field defined on an open subset of $\mathbb{R}^3$ and taking values in $\mathbb{S}^2$. Then
\begin{equation}\label{k1}
(\dive \n)^2 = | (\nn) \dive(\nn)|^2.
\end{equation}
\end{proposition}
\begin{proof}
Let us first note that
\begin{equation}\label{modn}
\color{black}0=\frac{1}{2} \nabla (|\n|^2) = \nabla \n^T \n,
\end{equation}
which follows from the fact that $|\n|=1$ everywhere. We now use \eqref{modn} to write{\color{black}
\begin{align}\notag
(\dive \n) \n &= (\dive \n)\n+\nabla \n^T \n   \\ \notag
&=  (\dive \n)\n+( \n\cdot\nabla \n^T \n ) \n  \\ \notag
&=  (\dive \n)\n+(\nabla \n\, \n\cdot \n ) \n  \\ \notag
&= (\nn)(\dive \n )\n+(\nn)(\nabla \n \, \n) \\ \label{divn}
&= (\nn)\dive (\nn).
\end{align}}
Taking $| \cdot |^2$ on both sides yields \eqref{k1}.
\end{proof}

\begin{proposition}\label{K2}
Let $\n$ be a smooth vector field defined on an open subset  of $\mathbb{R}^3$ and taking values in $\mathbb{S}^2$. Then
\begin{equation}\label{k2}
((\curl \n) \cdot \n)^2 = | (\n \otimes \n) \curl (\n \otimes \n)|^2,
\end{equation}
where for any matrix $M$, $| M |^2$ denotes the sum of the squares of the entries.
\end{proposition}
\begin{proof}
Let us first record
\begin{equation}\label{prel}
(\nn) (\nabla n_j \times \n)=(\n \cdot (\nabla n_j \times \n))\n = (\nabla n_j \cdot (\n \times \n))\n = 0.
\end{equation}
In the following calculation, we will use the fact that $|\n|^2=1$ in the first and third lines and use \eqref{prel} once to add 0 in the fourth line. We write
\begin{align}
\notag   ((\curl \n) \cdot \n)^2
&= |((\curl \n) \cdot \n) \n |^2 \notag\\
&= |(\nn) (\curl \n)|^2 \notag\\
&= \sum_j |n_j(\nn) (\curl \n)|^2 \notag\\
&= \sum_j |n_j(\nn) (\curl \n) + (\nn )(\nabla n_j \times \n)|^2 \notag\\
&= \sum_j |(\nn) (n_j\curl \n +  \nabla n_j \times \n)|^2 \notag\\
&=\sum_j |(\nn) \curl (n_j \n)|^2.\label{long1}
\end{align}
But $ \curl (n_j \n)$ is precisely the $j$-th column of $\curl (\n \otimes \n)$, so that 
$$
\sum_j |(\nn) \curl (n_j \n)|^2 = |(\nn)\curl (\nn)|^2.
$$
Combining this with \eqref{long1} finishes the proof of \eqref{k2}.
\end{proof}

\begin{proposition}\label{K3}
Let $\n$ be a smooth vector field defined on an open subset  of $\mathbb{R}^3$ and taking values in $\mathbb{S}^2$. Then
\begin{equation}\label{k3}
|(\curl \n) \times \n|^2 = |(\id - \nn) \dive (\nn)|^2
\end{equation}
\end{proposition}
\begin{proof}
Using the calculation from \eqref{divn} of $(\nn) \dive (\nn)$, let us first write{\color{black}
\begin{align*}
(\id - \nn)\dive (\nn) &= \dive(\nn) - (\dive \n)\n \\
&= \nabla \n\,\n+(\dive \n)\n - (\dive \n)\n\\
&= \nabla \n\,\n.
\end{align*}
Now recalling \eqref{modn}, we may subtract $\nabla \n^T \n=0$ from the right hand side of previous equation to obtain
\begin{align}\notag
(\id - \nn)\dive (\nn) &= \nabla \n\,\n - \nabla \n ^T \n \\ \notag
&= \begin{pmatrix} 0 & n_{1,y}-n_{2,x} & n_{1,z}-n_{3,x} \\ n_{2,x}-n_{1,y} & 0 & n_{2,z}-n_{3,y} \\ n_{3,x}-n_{1,z} & n_{3,y}-n_{2,z} & 0 \end{pmatrix}\n \\ \notag
&=  \begin{pmatrix} 0 & -(\curl \n)_3 & (\curl \n)_2 \\ (\curl \n)_3 & 0 & -(\curl \n)_1 \\ -(\curl \n)_2 & (\curl \n)_1 & 0 \end{pmatrix} \n \\ \notag
&= \begin{pmatrix}   (\curl \n)_2n_3-(\curl \n)_3n_2  \\ (\curl \n)_3n_1  -(\curl \n)_1n_3 \\ (\curl \n)_1n_2  -(\curl \n)_2n_1 \end{pmatrix} \\ \notag
&= (\curl \n) \times \n.
\end{align}}
Taking $|\cdot|^2$ on both sides completes the proof.
\end{proof}

\begin{proposition}\label{equalK}
Let $\n$ be a smooth vector field defined on an open subset  of $\mathbb{R}^3$ and taking values in $\mathbb{S}^2$. Then
\begin{equation}\label{equalk}
 |\nabla \n|^2=| (\id - \nn) \curl (\nn)|^2.
\end{equation}
\begin{proof}
Let us first calculate $|\curl(\nn)|^2$, after which we can use Proposition \ref{K2} to find $| (\id - \nn) \curl (\nn)|^2 $. Invoking Lemma \ref{tensorcurl} and then expanding, we write
\begin{align}\notag
|\curl(\nn)|^2 &= \sum_j | \curl (n_j\n)|^2 \\ \notag
&= \sum_j | \nabla n_j \times \n + n_j (\curl \n)|^2 \\ \notag
&= \sum_j |\nabla n_j \times \n|^2 + \sum_j n_j^2|\curl \n|^2 + \sum_j 2 (\nabla n_j \times \n) \cdot ( n_j \curl \n)\\ \label{123}
&=: I+II+III.
\end{align}
For $I$, we use Lagrange's identity and the identity $(\curl \n)\times \n=\nabla \n \, \n$ from the previous lemma to write
\begin{align}\notag
I&= \sum_j (|\nabla \n_j|^2 |\n|^2 - (\nabla \n_j \cdot \n)^2)\\ \notag
&= \left(\sum_j |\nabla \n_j|^2\right)-|\nabla \n \,\n|^2 \\ \label{I}
&= |\nabla \n|^2 - |(\curl \n) \times \n|^2.
\end{align}
Moving on to $II$, we immediately see that
\begin{equation}\label{II}
II = |\curl \n|^2.
\end{equation}
Finally, $III$ vanishes since
\begin{equation}\label{III}
III= \sum_j  (\nabla (|n_j|^2) \times \n) \cdot (  \curl \n) = ( \nabla (|\n|^2) \times \n) \cdot (\curl \n)=0.
\end{equation}
Substituting \eqref{I}-\eqref{III} into \eqref{123} yields
\begin{equation}\label{curlnn}
|\curl (\nn)|^2 = |\nabla \n|^2 -  |(\curl \n) \times \n|^2 + |\curl \n|^2 = |\nabla \n|^2 + ((\curl \n) \cdot \n)^2.
\end{equation}
But with the aid of Lemma \ref{K2}, we can also calculate $|\curl (\nn)|^2$ as
\begin{align}\notag
|(\id - \nn)\curl (\nn)&|^2 + |(\nn)\curl(\nn)|^2 \\ \label{alt} &= |(\id - \nn)\curl (\nn)|^2 + ((\curl \n)\cdot \n)^2.
\end{align}
Equating \eqref{curlnn} and \eqref{alt} equal and subtracting $((\curl \n)\cdot \n)^2$, we arrive at \eqref{equalk}.
\end{proof}
\end{proposition}
For uniaxial $Q$ such that $W(Q)=0$, we can use the preceding propositions to establish an equality between $\mathcal{F}_{LdG}$ and the Oseen-Frank energy $F_{OF}$.
\begin{proposition}\label{reduction}
Let $\Omega$ be an open set in $\mathbb{R}^n$ and suppose $W(Q)=0$ and $Q\in H^1(\Omega;\mathcal{N})$ is orientable, so that $Q=s_0(\nn-\id /3)$ for $n \in H^1(\Omega;\mathbb{S}^2)$. Then we have the equivalence
\begin{align}\notag
\mathcal{F}_{LdG}(Q)=\int_\Omega &\left(\frac{L_1}{2} \left|\left(\frac{s_0}{3}\id + Q\right)\dive Q\right|^2+ \frac{L_2}{2} \left|\left(\frac{s_0}{3}\id + Q\right) \curl Q\right|^2\right.\\\notag
&\left. +\frac{L_3}{2}\left|\left(\frac{2s_0}{3}\id - Q\right) \dive Q \right|^2 +\frac{L_4}{2}\left|\left(\frac{2s_0}{3}\id-Q\right)\curl Q\right|^2\right)\,dx \\* \notag
&\quad=\int_\Omega\left(\frac{K_1}{2}(\dive n)^2 + \frac{K_2}{2}((\curl n) \cdot n)^2 +\frac{K_3}{2}|(\curl n )\times n|^2\right. \\* 
&\left. \quad\quad\quad +\frac{K_2+K_4}{2}(\textup{tr}\, (\nabla n)^2 - (\dive n)^2)\right)\,dx =F_{OF}(n),\label{reductioneq}
\end{align}
where $s_0^4L_4=K_2+K_4$ and $s_0^4(L_i+L_4)=K_i$ for $1 \leq i \leq 3$.
\end{proposition}
\begin{proof}
Rearranging $Q=s_0(\nn-\id/3)$, we arrive at
\begin{equation}\label{rearrange}
    s_0\nn =  \frac{s_0}{3}\id + Q\textup{ and } s_0(\id - \nn) = \frac{2s_0}{3}\id - Q.
\end{equation}
Substituting the equations in \eqref{rearrange} into the left hand side of \eqref{reductioneq} and using the equalities \eqref{k1}, \eqref{k2}, and \eqref{k3}-\eqref{equalk} yields
\begin{align}\notag
&\int_\Omega \left(\frac{L_1}{2} \left|\left(\frac{s_0}{3}\id + Q\right)\dive Q\right|^2+ \frac{L_2}{2} \left|\left(\frac{s_0}{3}\id + Q \right) \curl Q \right|^2\right.\\\notag
&\left.\quad +\frac{L_3}{2}\left|\left(\frac{2s_0}{3}\id - Q\right) \dive Q \right|^2 +\frac{L_4}{2}\left|\left(\frac{2s_0}{3}\id-Q\right)\curl Q \right|^2\right)\,dx \\* \notag
=&\int_\Omega\left(\frac{s_0^4L_1}{2}|(\nn)\dive (\nn)|^2 + \frac{s_0^4L_2}{2}|(\nn)\curl (\nn)|^2\right. \\* \notag
&\quad\left.+\frac{s_0^4L_3}{2}|(\id-\nn)\dive (\nn )|^2 +\frac{s_0^4L_4}{2}|(\id - \nn) \curl(\nn)|^2\right)\,dx\\ \notag
=&\int_\Omega\left(\frac{s_0^4L_1}{2}(\dive n)^2 + \frac{s_0^4L_2}{2}((\curl n) \cdot n)^2+\frac{s_0^4L_3}{2}|(\curl n) \times n|^2\right. \\ \notag &\quad\quad\quad \left.+\frac{s_0^4L_4}{2}|\nabla n |^2\right)\,dx.
\end{align}
Recalling the identity 
\begin{equation}\notag
    |\nabla n|^2 = (\dive n)^2 + ((\curl n) \cdot n)^2 + |(\curl n) \times n|^2 + \tr(\nabla n)^2 - (\dive n)^2
\end{equation}
for smooth $n: \Omega \to \mathbb{S}^2$, we can rewrite the last integral as
\begin{align*}
    \int_\Omega&\left(\frac{s_0^4(L_1+L_4)}{2}(\dive n)^2 + \frac{s_0^4(L_2+L_4)}{2}((\curl n) \cdot n)^2\right.\\ &\left.\quad+\frac{s_0^4(L_3+L_4)}{2}|(\curl n) \times n|^2+\frac{s_0^4L_4}{2}\left(\tr(\nabla n)^2 - (\dive n)^2\right)\right)\,dx.
\end{align*}\end{proof}
We will refer to the elastic energy density terms in $\mathcal{F}_{LdG}$ as
\begin{align}\notag
\sigma&(Q):=\frac{L_1}{2} \left|\left(\frac{s_0}{3}\id + Q\right)\dive Q\right|^2+ \frac{L_2}{2} \left|\left(\frac{s_0}{3}\id + Q\right) \curl Q\right|^2 \\ \label{sigma}
&+\frac{L_3}{2}\left|\left(\frac{2s_0}{3}\id - Q\right) \dive Q \right|^2 +\frac{L_4}{2}\left|\left(\frac{2s_0}{3}\id-Q\right)\curl Q\right|^2 .
\end{align}

\begin{remark}
It is straightforward that $\sigma$ satisfies the the requisite frame indifference and material symmetry conditions; {cf. Lemma \ref{term35inv}}. Moreover, {in the Appendix} we identify each term in $\sigma$ as an appropriately weighted sum of terms from the generalized Landau-de Gennes theory \cite{LoMoTr87}. For example, the term $\left|\left(\frac{s_0}{3}\id + Q\right)\dive Q \right|^2$ corresponds to the  $L_2^{(2)}$-, $L_3^{(3)}$-, and $L_6^{(4)}$-invariants.
\end{remark}
\begin{remark}
{To ensure that $\sigma$ is non-negative, it is natural to require that each $L_i$ is non-negative. In fact, as we will see in the upcoming Proposition \ref{Lbdh1}, when each $L_i$ is positive, $\mathcal{F}_{LdG}$ is coercive over $H^1$. Conversely, if one of the $L_i$'s is negative, then one can construct a map $Q$ such that $\sigma(Q) < 0.$ For example, if $L_1<0,$ a curl-free map $Q$ can be constructed which satisfies 
$$
\frac{|L_1|}{2} \left|\left(\frac{s_0}{3}\id + Q\right)\dive Q\right|^2>\frac{L_3}{2}\left|\left(\frac{2s_0}{3}\id - Q\right) \dive Q \right|^2
$$
pointwise and ensures that $\sigma(Q)<0$.} 
\end{remark}
\begin{remark}{The Oseen-Frank elastic constants $K_i$ and our constants $L_i$ are related to each other as follows:
$$L_4=\displaystyle\frac{K_2+K_4}{s_0^4},\textup{ and }L_i=\displaystyle\frac{K_i-K_2-K_4}{s_0^4}\textup{ for }1 \leq i \leq 3,$$
or
$$
K_i = s_0^4(L_i+L_4)\textup{ for }1\leq i \leq 3,\textup{ and }K_4 =- s_0^4L_2.
$$
In the equal elastic constant case, $K_1=K_2=K_3=K$, $K_4=0$, so that $$F_{OF}(n)=\frac{K}{2}\int |\nabla n|^2/2,$$ each $L_i$ vanishes except for $L_4$. If we include the additional term $\int_\Omega L_5|\nabla Q|^2/2$, then the equal elastic constant case yields a $Q$-tensor model in which only $L_4$- and $L_5$-terms survive, satisfying $L_4+2L_5/s_0^2 = K$. As a special case, one may choose $L_4=0$, so that the elastic energy is just the Dirichlet energy for the $Q$-tensor.}
\end{remark}
\begin{remark}
One might inquire as to the relationship between the inequalities $L_i \geq 0$ and Ericksen's inequalities
\begin{equation}\label{ericksen}
    2K_1 \geq K_2+K_4,\, K_2 \geq |K_4|,\, K_3\geq 0
\end{equation}
for $F_{OF}$, which guarantee that the energy density in $F_{OF}$ is non-negative \cite{Er66}. It is quickly checked that if each $L_i$ is non-negative, then \eqref{ericksen} is satisfied. Conversely, if $\{K_i\}$ satisfy \eqref{ericksen}, it can be checked that the additional assumptions
$$
K_1  \geq K_2+K_4,\, K_3 \geq K_2+K_4,\, K_4 \leq 0
$$
are needed so that each $L_i$ is non-negative. It is possible that these additional assumptions can be relaxed through the inclusion of more quartic terms identified by \cite{LoMoTr87}, but we do not pursue this issue further.
\end{remark}

\section{Analysis and $\Gamma$-Convergence of $\mathcal{F}_{LdG}$}\label{ourmodel}
Motivated by Proposition \ref{reduction}, we combine the quartic $Q$-tensor elastic terms with the bulk potential $W$ given by \eqref{LdG} to obtain the following generalized Landau-de Gennes energy with quartic elastic energy density, defined over $Q\in H^1(\Omega;\mathcal{S})$:
\begin{align*}\notag
\mathcal{F}_{LdG}(Q) := \int_\Omega &\left(\frac{L_1}{2} \left|\left(\frac{s_0}{3}\id + Q\right)\dive Q\right|^2+ \frac{L_2}{2} \left|\left(\frac{s_0}{3}\id + Q \right) \curl Q\right|^2\right.\\*
&\left. +\frac{L_3}{2}\left|\left(\frac{2s_0}{3}\id - Q\right) \dive Q \right|^2 +\frac{L_4}{2}\left|\left(\frac{2s_0}{3}\id-Q\right)\curl Q\right|^2\right.\\*
&\qquad +  W(Q)\bigg)\,dx.
\end{align*}
We assume that $L_i >0$ for each $i.$ This will ensure that our energy is coercive over $H^1(\Omega;\mathcal{S})$; cf. Proposition \ref{Lbdh1}. Let us point out that $Q$ merely belonging to $H^1(\Omega;\mathcal{S})$ is not enough to conclude that $\mathcal{F}_{LdG}(Q) < \infty$. If, for example, $Q\in (L^\infty \cap H^1)(\Omega;\mathcal{S})$ or $W^{1,4}(\Omega;\mathcal{S})$, however, then $\mathcal{F}_{LdG}$ is necessarily finite.\par
In order to establish a meaningful asymptotic connection between this generalizaed Landau-de Gennes model and Oseen-Frank, we next non-dimensionalize the energy $\mathcal{F}_{LdG}$ by scaling the spatial coordinates
\begin{align*}
    \overline{x}=x/D,
\end{align*}
where $D:=\mathrm{diam}(\Omega)$. We also rescale $Q$ by letting 
\begin{align*}
    \overline{Q}(\overline{x}):= Q(x)/s_0.
\end{align*}
Although the order parameter $Q$ is already dimensionless, dividing by the dimensionless parameter $s_0$ enables us to eliminate it from the elastic energy density. Let $\zeta:= L_1s_0^4/D^2 $, and define a dimensionless elastic energy density via
\begin{align*}
    \overline{\sigma}\left(\overline{Q}\right) := \frac{1}{\zeta }\sigma(Q).
\end{align*}
Then, in the $\overline{x}$-coordinates,
\begin{align}\notag
    \overline{\sigma}\left(\overline{Q}\right) &=\frac{\overline{L}_1}{2} \left|\left(\id/3 + \overline{Q}\right)\dive \overline{Q}\right|^2+ \frac{\overline{L}_2}{2} \left|\left(\id/3 + \overline{Q}\right) \curl \overline{Q}\right|^2 \\* \label{rescalesigma}
&+\frac{\overline{L}_3}{2}\left|\left(2\id/3 - \overline{Q}\right) \dive \overline{Q} \right|^2 +\frac{\overline{L}_4}{2}\left|\left(2\id/3-\overline{Q}\right)\curl \overline{Q}\right|^2 ,
\end{align}
where $\overline{L}_i:=L_i/L_1$ for $1\leq i \leq 4$. Next, we rescale the Landau-de Gennes potential $W$ by introducing
\begin{align}\label{Wbar}
    \overline{W}\left(\overline{Q}\right) := \overline{a}\,\mathrm{tr}\left(\overline{Q}^2\right)-\frac{2\overline{b}}{3}\,\mathrm{tr}\left(\overline{Q}^3\right)+\overline{c}\left(\mathrm{tr}\left(\overline{Q}^2\right)\right)^2,
\end{align}
where $\overline{a}=as_0^2/c$, $\overline{b}=bs_0^3/c$, and $\overline{c}=s_0^4/2$. Then setting $\varepsilon:=\sqrt{\frac{\zeta}{c}}$, we have
\begin{align*}
    \overline{W}\left(\overline{Q}\right) = \frac{\varepsilon^2}{\zeta}W(Q).
\end{align*}
Recall that we are assuming the global minimum of $\overline{W}$ is 0, and it is now achieved along the set
\begin{equation}\label{newn}
    \left\{\nn - \frac{1}{3}\id: n \in \mathbb{S}^2 \right\},
\end{equation}
which we will still refer to as $\mathcal{N}$. If we write $\overline{\Omega}:=\Omega/D$, then the total energy is given by
\begin{align*}\notag
\int_{\Omega} &\left(\sigma(Q) +  W(Q)\right)\,dx= \zeta D^3\int_{\overline{\Omega}} \left(\overline{\sigma}\left(\overline{Q}\right) +  \frac{1}{\varepsilon^2}\overline{W}\left(\overline{Q}\right)\right)\,d\overline{x}.
\end{align*}
Finally, noting that $\zeta D^3$ has the dimensions of energy, we define the non-dimensionalized energy via
\begin{equation*}
\mathcal{F}_{\varepsilon}\left(\overline{Q}\right) := \int_{\overline{\Omega}} \left(\overline{\sigma}\left(\overline{Q}\right) +  \frac{1}{\varepsilon^2}\overline{W}\left(\overline{Q}\right)\right)\,d\overline{x}.
\end{equation*}
The parameter $\varepsilon$ can be interpreted as $\xi_{NI}/D$, where $\xi_{NI}$ is the nematic correlation length that determines isotropic core size \cite{gartland}. With a slight abuse of notation, we will dispose of the bars for the rest of the paper, so that the non-dimensionalized energy is written as
\begin{align}\label{Feps}
\mathcal{F}_{\varepsilon}(Q) =  \int_{\Omega} &\left(\sigma(Q) +  \frac{1}{\varepsilon^2}W(Q)\right)\,dx,
\end{align}
with $\sigma(Q)$ henceforth given by \eqref{rescalesigma} and $W(Q)$ given by \eqref{Wbar}.
We will consider the $\varepsilon \to 0$ limit for $\mathcal{F}_\varepsilon$, which should be understood as the limit in which the nematic correlation length is vanishingly small compared to the size of the domain. \par
Finally, we will require throughout the rest of the article that admissible competitors for $\mathcal{F}_{\varepsilon}$ are subject to a Dirichlet boundary condition
$$
 g \in \textup{Lip}\,(\partial \Omega ;\mathcal{S}).
 $$
 For such $g$, the set $H^1_g(\Omega;\mathcal{S})$ contains competitors which are also Lipschitz, so that the infimum of $\mathcal{F}_{\varepsilon}$ among $H^1_g(\Omega;\mathcal{S})$ is not infinity. Later on, when we are considering the asymptotic behavior of $\mathcal{F}_{\varepsilon}$ as $\varepsilon \to 0$, we will further require that $g$ takes values in $\mathcal{N}$.\par 
\begin{proposition}[Coercivity of $\mathcal{F}_{\varepsilon}$]\label{Lbdh1}There exist constants $C_1=C_1(\{L_i\})>0$ and $C_2=C_2(g)$ such that for any $Q \in H^1_g(\Omega;\mathcal{S})$,
\begin{equation}\label{lbdh1}
\| Q \|_{H^1}^2 \leq C_1\mathcal{F}_{\varepsilon}(Q)+C_2.
\end{equation}
\end{proposition}
\begin{proof}
It is standard to bound $\|Q\|_{L^2}^2$ from above by the potential term in the energy $\mathcal{F}_{\varepsilon}$, so we focus on $\sigma$. First, recall that for any symmetric matrix $M$ and vectors $v_1,v_2$,
$$
\langle Mv_1, v_2 \rangle = \langle v_1, Mv_2 \rangle.
$$
Using this fact, we calculate 
\begin{align}\notag
\left|(\id/3 + Q)\dive Q \right|^2+&\left| (2\id/3-Q)\dive Q\right|^2\\ \notag &= \langle \dive Q , (\id/3 + Q)^2 \dive Q\rangle + \langle \dive Q, (2\id/3 - Q)^2 \dive Q \rangle \\* \notag
&= \langle \dive Q , (2Q^2- 2Q/3+5\id/9)\dive Q \rangle  \\ \notag
&= \left\langle \dive Q , \left(\id/2 + 2\left(Q - \id/6\right)^2\right)\dive Q \right\rangle \\* \label{keydiv}
&= |\dive Q|^2/2 + 2\left|\left(Q - \id / 6\right) \dive Q \right|^2 
\end{align}
Similarly, we have
\begin{align}\notag
\left|(\id/3 + Q)\curl Q \right|^2+&\left| (2\id/3-Q)\curl Q\right|^2\\* \label{keycurl} &=  |\curl Q|^2/2 + 2\left|\left(Q - \id / 6\right) \curl Q \right|^2 
\end{align}
The equalities \eqref{keydiv}-\eqref{keycurl} imply that
\begin{equation}\label{eles}
\frac{1}{2} \min_i \frac{L_i}{2}  \left(\left|\dive Q \right|^2+\left|\curl Q \right|^2\right) \leq \sigma(Q) .
\end{equation}
To bound this from below using $|\nabla Q|^2$, we need the identity
\begin{equation}\label{combine}
|\nabla Q|^2 = \sum_j  \left((\dive Q_j)^2 + |\curl Q_j|^2+ \textup{tr}\, (\nabla Q_j )^2- (\dive Q_j )^2\right)
\end{equation}
 cf. \cite[Lemma 1.4]{HaKiLi86}, where $Q_j$ is the $j$-th row of $Q$. Therefore, with $C_1=4/(\min_{1\leq i \leq 4 } L_i)$ we can combine \eqref{eles} and \eqref{combine} to arrive at
\begin{equation*}
\int_\Omega | \nabla Q|^2 \, dx \leq C_1 \int_\Omega \sigma(Q)\,dx + \int_\Omega \sum_j \left(\textup{tr}\, (\nabla Q_j )^2- (\dive Q_j )^2\right)\, dx.
\end{equation*}
The remainder on the right hand side is a null Lagrangian, and since the boundary data $g$ is Lipschitz, it can be written as
\begin{align}
\int_{\partial\Omega}\left(\sum_j(\nabla_{\mathrm{tan}} g_j) g_j - (\tr \nabla_{\mathrm{tan}} g_j)g_j\right)\cdot \nu \, d\mathcal{H}^2,\label{nul}
\end{align}
which is a constant $C_2$ independent of $Q\in H^1_g(\Omega;\mathcal{S})$.
We refer the reader to \cite[Lemma 1.2]{HaKiLi86} for the derivation of this formula.

\end{proof}\par
To prove the existence of minimizers of $\mathcal{F}_{\varepsilon}$ among $H^1_g(\Omega;\mathcal{S})$ and to prove a $\Gamma$-convergence result, we will need the following proposition.
\begin{proposition}[Lower-semicontinuity of $\sigma$]\label{4.0Lsc}
For any sequence $Q_n$ which converges weakly in $H^1(\Omega;\mathcal{S})$ to $Q \in H^1(\Omega;\mathcal{S})$, we have
\begin{equation}\label{lsceq}
\int_\Omega \sigma(Q) \,dx \leq \liminf_{n \to \infty} \int_{\Omega} \sigma(Q_n) \, dx.
\end{equation}
\end{proposition}
\begin{proof}
We focus on the term $|(\id/3 + Q_n)\dive Q_n|^2$; the argument for the other terms is the same. Let us assume that the right hand side of \eqref{lsceq} is finite; if it is not, the proof is trivial. The essence of the subsequent proof is the real analysis fact
\begin{align*}
f_n \to f \textup{ strongly in }L^2, &\textup{ }g_n \to g \textup{ weakly in } L^2, \textup{ }\|f_n g_n \|_{L^2} \leq C < \infty\\* 
&\implies f_ng_n \to f g \textup{ weakly in }L^2 . 
\end{align*}\par
Without loss of generality, we can assume (by restricting to a subsequence) that $\liminf \int_\Omega |(\id/3+Q_n)\dive Q_n|^2 \, dx $ is finite and the sequence of integrals converges to its limit inferior. Let us first recall that weak convergence in $H^1(\Omega;\mathcal{S})$ entails strong convergence in $L^2(\Omega;\mathcal{S})$, so that 
\begin{equation}\label{strong}
\id/3+Q_n \to \id/3+Q \textup{ in }L^2(\Omega;\mathcal{S}).
\end{equation}
Next, since $\dive Q_n$ converges weakly in $L^2(\Omega;\R^3)$ to $\dive Q$ and $\id/3 + Q_n$ converges strongly in $L^2(\Omega;\mathcal{S})$ to $\id/3 + Q$, we have for any $\phi \in L^\infty(\Omega;\mathbb{R}^3)$:
\begin{align*}
\int_\Omega (\id/3+&Q_n)\dive Q_n \cdot \phi\\* &=\int_\Omega \left[(\id/3 +  Q_n)-(\id/3+ Q)\right]\dive Q_n \cdot \phi +\int_\Omega (\id/3 + Q)\dive Q_n \cdot \phi\\
&\to \int_\Omega (\id/3+Q)\dive Q \cdot \phi .
\end{align*}
Thus $$(\id/3 + Q_n)\dive Q_n \to (\id/3 + Q)\dive Q\textup{ weakly in }L^1(\Omega;\R^3).$$ Now from the uniform $L^2$ bound on $(\id/3+ Q_n)\dive Q_n$, we have that (up to a subsequence) $$(\id/3 + Q_n)\dive Q_n \to h \textup{ weakly in }L^2(\Omega;\mathbb{R}^3)$$ for some $h \in L^2(\Omega;\mathbb{R}^3)$. But from the previous observation and the uniqueness of weak limits, we deduce that the weak $L^2$-limit $h$ must coincide with $(\id/3 + Q)\dive Q $, the weak $L^1$-limit.  The inequality
\begin{equation*}
\int_\Omega |(\id/3 + Q)\dive Q|^2 \, dx \leq \liminf_{n \to \infty} \int_\Omega |(\id/3 + Q_n)\dive Q_n|^2 \, dx
\end{equation*}
now follows from the lower semicontinuity of the $L^2$-norm under weak convergence. Repeating the same argument for the other terms in $\sigma$ concludes the proof.
\end{proof}
{
\begin{remark}
The proposition \ref{4.0Lsc} should hold for other quartic elastic energies, provided the energy is strongly lower-semicontinuous and convex  \cite{DaGa98}. This would also allow for generalization of the $\Gamma$-convergence result (Theorem \ref{tt1}) to other models with quartic elastic energies.
\end{remark}
}
Next, we turn our attention to the existence of minimizers of $\mathcal{F}_{\varepsilon}$. 
\begin{theorem}[Existence of a minimizer]\label{existence}
For any $\varepsilon>0$ and Lipschitz $g:\partial \Omega \to \mathcal{S}$, there exists $Q_0$ which minimizes $\mathcal{F}_{\varepsilon}$ within $H^1_g(\Omega;\mathcal{S})$.
\end{theorem}
\begin{proof}
Fix $\varepsilon>0$ and $g$ as stated in the theorem. By virtue of the previous two propositions, the existence of a minimizer will follow without difficulty from the direct method in the calculus of variations.\par
Let $\{ Q_n\}$ be a sequence such that
\begin{equation*}
\lim_{n\to \infty} \mathcal{F}_{\varepsilon}(Q_n) = \inf \{\mathcal{F}_{\varepsilon}(Q): Q \in H^1_g(\Omega;\mathcal{S})  \}.
\end{equation*}
As noted earlier, since $g$ is Lipschitz, the infimum is not $\infty$. Then by Proposition \ref{Lbdh1}, we have a uniform $H^1$ bound on $\{ Q_n \}$ and a subsequence, which we still refer to as $\{ Q_n\}$, converging weakly in $H^1$ to some $Q_0\in H^1_g(\Omega;\mathcal{S})$. Proposition \ref{4.0Lsc} then yields
\begin{equation}\label{elas}
\int_\Omega \sigma(Q_0) \, dx \leq \liminf_{n\to \infty} \int_\Omega \sigma(Q_n)\,dx.
\end{equation}
By Rellich's theorem, we may assume as well that $Q_n$ converge in $L^4$ to $Q_0$, from which we deduce
\begin{equation}\label{pot}
\int_\Omega W(Q_0) \, dx = \lim_{n \to \infty}\int_\Omega W(Q_n)\,dx.
\end{equation}
The minimality of $Q_0$ is now a consequence of \eqref{elas} and \eqref{pot}.
\end{proof}
We are now interested in the asymptotic behavior of minimizers of $\mathcal{F}_{\varepsilon}$ as $\varepsilon\to 0$. Let us begin by identifying a limiting functional. In the limit $\varepsilon \to 0$, it is clear that competitors with finite energy will have to take values in $\mathcal{N}$, the well of $W$, cf. \eqref{newn}. Assume that in addition to being Lipschitz, the boundary data $g$ takes values in $\mathcal{N}$, so that $W(g)=0$. The set of satisfactory boundary data includes, for example, $g$ which are formed using a Lipschitz vector field $n:\partial \Omega \to \mathbb{S}^2$ and considering the tensor field
\begin{equation}
n\otimes n - \frac{1}{3}\id :\partial \Omega \to \mathcal{N},
\end{equation}
cf. \cite[Lemma 1.1]{HaKiLi86}. The limiting functional $\mathcal{F}_{0}$ is then defined by 
\begin{equation}\notag
\mathcal{F}_{0}(Q)= \begin{cases} 
      \displaystyle\int_\Omega \sigma(Q)\,dx & \textup{if }Q \in H^1_g(\Omega;\mathcal{N}), \\
      \infty & \textup{otherwise,}
   \end{cases}
\end{equation} with $\sigma$ given by \eqref{rescalesigma}.\par
Let us point out a key feature of the limiting model: $\mathcal{F}_{0}$ coincides with the Oseen-Frank model, in the sense of Proposition \ref{reduction}.
The question of when minimizing $\mathcal{F}_0$ among $Q$-tensor fields coincides with minimizing the version of the Oseen-Frank energy above is more delicate. There may well be strictly more competitors in the space of $Q$-tensors than in the space of $\mathbb{S}^2$-valued fields $n$ due to the possible `non-orientability' of a $Q$-tensor field. It has been shown in \cite{BaZa11} that when $\Omega \subset \mathbb{R}^k$ is simply connected and $k=2,3$, every $Q \in H^1(\Omega;\mathcal{N})$ has a lifting $n^Q\in H^1(\Omega;\mathbb{S}^2)$ such that 
\begin{equation*}
Q= n^Q\otimes n^Q-\id/3  .
\end{equation*}
If $\Omega$ is not simply connected or if $p<2$, then there might exist tensor fields which cannot be `oriented' to produce a globally defined corresponding director. We refer the reader to \cite{BaZa11} for a more detailed treatment.\par
We state the first of two theorems regarding the asymptotic behavior of $\mathcal{F}_{\varepsilon}$ and its minimizers.
\begin{theorem}[$\Gamma$-convergence]\label{tt1}
For any choice of boundary data $g$ as above, the sequence $\{\mathcal{F}_{\varepsilon}\}$ $\Gamma$--converges in the weak topology of $H^1_g(\Omega;\mathcal{S})$ to $\mathcal{F}_0$. That is, 
\begin{enumerate}
\item For any $Q\in H^1_g(\Omega;\mathcal{S})$ and for any sequence $\{Q_\varepsilon\}$ in $H^1_g(\Omega;\mathcal{S})$, 
\begin{equation}\label{4.0lsceq}
	Q_\varepsilon \rightharpoonup Q \text{ in } H^1_g(\Omega;\mathcal{S}) \text{ implies } \liminf_{\varepsilon \rightarrow 0} \mathcal{F}_{\varepsilon}(Q_\varepsilon) \geq \mathcal{F}_0(Q),
\end{equation}
and
\item For each $Q\in H^1_g(\Omega;\mathcal{S})$ there exists a recovery sequence $\{Q_\varepsilon\}$ in $H^1_g(\Omega;\mathcal{S})$ satisfying
\begin{equation}\label{4.0reco1}
Q_\varepsilon \rightharpoonup Q_0 \text{ in } H^1_g(\Omega;\mathcal{S}),
\end{equation}
\begin{equation}\label{4.0reco2}
\lim_{\varepsilon\rightarrow 0} \mathcal{F}_{\varepsilon}(Q_\varepsilon)= \mathcal{F}_0(Q_0).
\end{equation}
\end{enumerate}
\end{theorem}
Before we present the proof, we state a compactness proposition, which follows immediately from Proposition \ref{Lbdh1}.
\begin{proposition}[Compactness]\label{compactness}
Let $\{Q_\varepsilon \}$ be a sequence of maps from $\Omega$ to $\mathcal{S}$, and assume that the sequence of energies $\mathcal{F}_{\varepsilon}(Q_\varepsilon)$ is uniformly bounded. Then there exists a subsequence $\{Q_{\varepsilon_j} \}$ and $Q \in H^1_g(\Omega;\mathcal{N})$ such that $Q_{\varepsilon_j} \rightharpoonup Q$ in $H^1(\Omega;\mathcal{S})$.
\end{proposition}
\begin{proof}[Proof of Theorem \ref{tt1}]
The lower-semicontinuity condition \eqref{4.0lsceq} has been proved in Proposition \ref{4.0Lsc}. For the construction of a recovery sequence given some $Q_0$, we can simply take $Q_\varepsilon=Q_0$ for all $\varepsilon$.
\end{proof}
Finally, we prove
\begin{theorem}[$H^1$-convergence of a subsequence of minimizers]\label{strongcon}
For any sequence of minimizers $\{ Q_\varepsilon \}$ of $\mathcal{F}_{\varepsilon}$, there exists $Q_0\in H^1_g(\Omega;\mathcal{N})$ which minimizes $\mathcal{F}_0$ and a subsequence $\{ Q_{\varepsilon_j}\}$ converging to $Q_0$ strongly in $H^1_g(\Omega;\mathcal{S})$. 
\end{theorem}
\begin{proof}
We will prove the theorem under the assumption that
\begin{equation}\label{assume}
\min_{1 \leq i \leq 4}L_i =L_1=1>0;
\end{equation}
the proof when one of the other $L_i$'s is the smallest is similar. Appealing to Proposition \ref{compactness} yields a subsequence $\{ Q_{\varepsilon_j}\}$, which we will call $\{Q_j\}$ for convenience, such that $Q_j$ converges weakly in $H^1(\Omega;\mathcal{S})$ to some $Q_0\in H^1_g(\Omega;\mathcal{N})$ of $\mathcal{F}_0$. It is a classical fact from the theory of $\Gamma$-convergence that $Q_0$ minimizes $\mathcal{F}_0$. To show the strong convergence of $Q_j$, we will use the fact that weak convergence together with convergence of norms implies strong convergence.\par
It will be necessary to first extend our maps to compactly supported Sobolev maps on a larger domain. Let $\Omega' \subset \R^3$ be a smooth domain such that $\Omega$ is compactly contained in $\Omega'$, and let $P: \Omega' \setminus \Omega \to \mathcal{S}$ be an $H^1$ function such that its trace on $\partial \Omega$ is $g$ and its trace on $\partial \Omega'$ is 0. We extend each $Q_{j}$ to a map $\tilde{Q}_{j} \in H_0^1(\Omega' ; \mathcal{S})$ via
\[  \tilde{Q}_{j}(x) = \begin{cases} 
      Q_{j}(x)  & x\in \Omega, \\
      P(x) &  x \in \Omega' \setminus \Omega,
   \end{cases}
\]
and similarly for $Q_0$. We will now use the calculations \eqref{keydiv}-\eqref{keycurl} from Proposition \ref{Lbdh1} to extract a $|\nabla Q|^2$ term from $\sigma$. First, from \eqref{keydiv}-\eqref{keycurl} we have
\begin{align*}
\notag
&\int_{\Omega'} \left(\frac{1}{2} |(\id/3 + \tilde{Q}_j)\dive \tilde{Q}_j|^2+ \frac{L_2}{2} \left|(\id/3+  \tilde{Q}_j) \curl \tilde{Q}_j\right|^2 \right.\\*
&\left.\quad\quad+\frac{L_3}{2}\left|(2\id/3 - \tilde{Q}_j) \dive \tilde{Q}_j \right|^2 +\frac{L_4}{2}|(2\id/3 -\tilde{Q}_j)\curl \tilde{Q}_j |^2 \right)\,dx\\
&\quad=\frac{1}{2}\int_{\Omega'} \left( |(\id/3 + \tilde{Q}_j)\dive \tilde{Q}_j|^2+  |(\id/3 + \tilde{Q}_j) \curl \tilde{Q}_j|^2 \right.\\*
&\left.\quad\quad+|(2\id/3 - \tilde{Q}_j) \dive \tilde{Q}_j |^2 +|(2\id/3 - \tilde{Q}_j)\curl \tilde{Q}_j |^2 \right.\\*
&\quad\quad + \left. \frac{L_2-1}{2} \left|(\id/3 + \tilde{Q}_j) \curl \tilde{Q}_j\right|^2+\frac{L_3-1}{2}\left|(2\id/3 - \tilde{Q}_j) \dive \tilde{Q}_j \right|^2 \right.\\
&\left.\quad\quad +\frac{L_4-1}{2}|(2\id/3-\tilde{Q}_j)\curl\tilde{Q}_j|^2 \right)\,dx\\
&\quad=\frac{1}{2}\int_{\Omega'} \left( \left(|\dive \tilde{Q}_j|^2 +|\curl \tilde{Q}_j|^2\right)/2 + 2|(\tilde{Q}_j- \id/6 )\dive\tilde{Q}_j|^2\right.\\
&\quad\quad \left.+ 2|(\tilde{Q}_j- \id/6 )\curl\tilde{Q}_j|^2+ \frac{L_2-1}{2} \left|(\id/3 + \tilde{Q}_j) \curl \tilde{Q}_j\right|^2 \right.\\
&\left.\quad\quad+\frac{L_3-1}{2}\left|(2\id/3 - \tilde{Q}_j) \dive \tilde{Q}_j \right|^2 +\frac{L_4-1}{2}|(2\id/3-\tilde{Q}_j)\curl\tilde{Q}_j|^2 \right)\,dx.
\end{align*}
But for any smooth $Q \in H_0^1(\Omega';\mathcal{S})$, we have by \eqref{combine} and \eqref{nul}
\begin{align}\notag
\int_{\Omega'} \left( |\dive Q|^2 + |\curl Q|^2 \right) \,dx &= \int_{\Omega'} \left(|\nabla Q|^2 + \textup{null Lagrangian}\right)\,dx\\ \label{byparts} &=\int_{\Omega'} |\nabla Q|^2 \,dx.
\end{align}
Hence \eqref{byparts} holds for all $Q \in H_0^1(\Omega' ; \mathcal{S})$ by density. Plugging this equality into our expression for $\int_{\Omega'} \sigma (\tilde{Q}_j)\,dx$ gives
\begin{align} \label{mess}
\int_{\Omega'}& \sigma(\tilde{Q}_j) \, dx \\* \notag
 &= \frac{1}{2}\int_{\Omega'} \left( |\nabla \tilde{Q}_j|^2/2 + 2|(\tilde{Q}_j-\id/6 )\dive\tilde{Q}_j|^2\right. \\ \notag
&\quad\quad \left.+ 2|(\tilde{Q}_j-\id/6 )\curl\tilde{Q}_j|^2 + \frac{L_2-1}{2} \left|(\id/3 + \tilde{Q}_j) \curl \tilde{Q}_j\right|^2 \right.\\ \notag
&\left.\quad\quad+\frac{L_3-1}{2}\left|(2\id/3 - \tilde{Q}_j) \dive \tilde{Q}_j \right|^2 +\frac{L_4-1}{2}|(2\id/3-\tilde{Q}_j)\curl\tilde{Q}_j|^2 \right)\,dx.
\end{align}
Define $\tilde{\sigma}(\tilde{Q}_j)$ to be the sum of the integrands of the right hand side of \eqref{mess}, so that
\begin{align*}
\int_{\Omega'}\sigma(\tilde{Q}_j)\,dx = \int_{\Omega'}\tilde{\sigma}(\tilde{Q}_j)\,dx.
\end{align*}
By the exact same argument as in Proposition \ref{4.0Lsc}, each of the individual terms in $\int_{\Omega'} \tilde{\sigma}(\tilde{Q}_j)\,dx$ is lower-semicontinuous with respect to weak $H^1$-convergence. Therefore,
\begin{align}
 \int_{\Omega'} \tilde{\sigma}(\tilde{Q}_0)\, dx \leq \liminf_{j \to \infty} \int_{\Omega'} \tilde{\sigma}(\tilde{Q}_j)\, dx. \label{half}
\end{align}
But since $Q_j$ minimizes $\mathcal{F}_{\varepsilon_j}$, we also have  
\begin{align}\notag
\limsup_{j \to \infty} \int_{\Omega'} \tilde{\sigma}(\tilde{Q}_j)\, dx 
&= \limsup_{j \to \infty} \int_{\Omega'} \sigma(\tilde{Q}_j)\, dx \\ \notag
&\leq  \limsup_{j \to \infty} \int_{\Omega'} \left( \sigma(\tilde{Q}_j)+ \frac{1}{\varepsilon_j} W(\tilde{Q}_j)\right)\, dx \\ \notag
& \leq \int_{\Omega'} \sigma(\tilde{Q}_0)\, dx \\
& = \int_{\Omega'} \tilde{\sigma}(\tilde{Q}_0)\, dx . \label{otherhalf}
\end{align}
Together, \eqref{half} and \eqref{otherhalf} give 
$$
\int_{\Omega'} \tilde{\sigma}(\tilde{Q}_0)\, dx =\lim_{j \to \infty} \int_{\Omega'} \tilde{\sigma}(\tilde{Q}_j)\, dx.
$$
Since each separate term in $\tilde{\sigma}$ is lower-semicontinuous, it must be the case that
\begin{equation}
\int_{\Omega'} |\nabla \tilde{Q}|^2=\lim_{j \to \infty} \int_{\Omega'} |\nabla \tilde{Q}_j|^2 \, dx.\label{normcon}
\end{equation}
From the weak $H^1$-convergence, we know that $\nabla \tilde{Q}_j \to \nabla \tilde{Q}_0$ weakly in $L^2$, which in conjunction with \eqref{normcon} implies that 
$$
\nabla \tilde{Q}_j \to \nabla \tilde{Q}_0 \textup{ in } L^2(\Omega' ;\mathcal{S})
.$$
Since $\tilde{Q}_j = P=Q_0$ on $\Omega' \setminus \Omega$, we have shown that in fact
$$
\nabla Q_j \to \nabla Q_0 \textup{ in } L^2(\Omega;\mathcal{S}),
$$
and the proof is complete under the assumption that $1=L_1\leq L_i$ for $2 \leq i \leq 4$. If $\min_i L_i$ is achieved by some $L_i$ where $i \neq 1$, the proof follows almost exactly as above, with the second integral in \eqref{mess} replaced by a similar expression involving the three largest $L_i$'s and their corresponding elastic terms.
\end{proof}

{\section{Discussion}
We have introduced a version of Landau-de Gennes model involving a quartic elastic energy that reduces to the four-constant Oseen-Frank energy for uniaxial $Q$-tensors. This version has several advantages, as laid out in the introduction, in addition to being capable of recovering the full Oseen-Frank energy. 

From the mathematical standpoint, it is far more stable than the frequently used cubic model since the energy is bounded from below and coercive. From the physical viewpoint, we have found this form of Landau-de Gennes particularly useful in modeling liquid crystalline states with co-existing isotropic and nematic phases, specifically in the context of lyotropic chromonic liquid crystals (LCLCs) \cite{kim2013morphogenesis,zhouetal}. To this end, in order to capture nematic-to-isotropic phase transitions, a $Q$-tensor has to vanish in a part of the physical domain while being in a uniaxial, $\mathbb{RP}^2$-valued nematic state elsewhere. Clearly, this setup falls outside of the scope of an $\mathbb{S}^2$-valued director theory such as that of Oseen and Frank. What is more, because in the Oseen-Frank limit the standard quadratic Landau-de Gennes model \eqref{elastic} reduces to a director model with two equal elastic constants, the quadratic Landau-de Gennes framework is not adequate when the contributions from splay, twist, and bend all incur different costs. In fact, it is well-known experimentally that for certain liquid crystals, including LCLCs, the elastic constants for splay, bend and twist can vary dramatically. It is thus a keen interest of ours to explore the effects of strong elastic disparity on evolution of interfaces and nematic singularities in isotropic-to-nematic transitions. We are encouraged by the results of gradient flow simulations using our model which exhibit good qualitative agreement with experiments, including the capture of such features as phase boundary singularities--the so-called ``boojums"--and vortex splitting, cf. \cite{GoKiLaNoSt19}. 

\begin{example}[Lyotropic Chromonic Liquid Crystals]\label{example}
Let us compare the standard quadratic Landau-de Gennes energy \eqref{elastic}, a cubic Landau-de Gennes energy \cite{BaMa10}, and the quartic energy proposed here in terms of modeling the LCLC's with highly disparate elastic constants considered in \cite{zhouetal}. The authors in \cite{zhouetal} measured the values of the Oseen-Frank elastic constants for disodium cromoglycate and found that the coefficients $K_1$ and $K_3$ for splay and bend are of the order 10 pN, while the twist coefficient is about 10 times smaller. However, the values of $K_1$ and $K_3$ are not equal. In Fig. 4 of \cite{zhouetal}, the ratio $K_1/K_3$ varies between approximately .4 and 1.2, depending on temperature and concentration. Therefore, modeling this liquid crystal using the standard quadratic model \eqref{elastic} is not feasible, since the values of $K_1$ and $K_3$ in terms of $L_i$ for this model are equal, cf. \cite[Eq. 4.17]{ball17}. If one were to attempt to use the cubic model with modified potential from \cite{BaMa10}, a necessary condition for minimization is that one of the coefficients, $L_1'$, is positive, cf. \cite[Eq. 5.16]{ball17}. In terms of the Oseen-Frank constants $K_i$ \cite[Equation 4.18]{ball17}, this condition is equivalent to
\begin{equation}\label{ineq}
(1-s_0)(K_1-K_3) + \displaystyle 3s_0\,K_2 > 0
\end{equation}
when $K_3>K_1$. If, for instance, one has $10 K_2 \approx K_1 \approx  K_3/2$ from the experiment and $s_0$ is in the typically observed interval $(0.6,0.7)$ \cite{ball2017liquid}, this inequality does not hold. Thus, the standard Landau-de Gennes model and its variants are ill-suited to modeling in this scenario, since they are not well-posed. In contrast, for the quartic model proposed here, the restrictions $K_1>K_2+K_4,$ $K_3 > K_2+K_4$ needed in order to guarantee coercivity do not present any obstacles. Furthermore, these inequalities are consistent with commonly observed values of $K_i$ \cite{garttalk}. 
\end{example}

Finally, we wish to emphasize that there are surely many, many other choices of elastic energy densities beside those we take in  \eqref{ourfldg} capable of recovering the four-constant Oseen-Frank energy for uniaxial $Q$. This should not cause concern since what is physically measurable are the Frank constants $K_1,\,K_2,\,K_3$ and $K_4$. Given that there are six different allowable cubic terms, 13 quartic terms, and even more terms if one ventures into higher order, it is clear that there are myriad combinations of terms in a $Q$-tensor based  
elastic energy density that can reduce to the four-constant Oseen-Frank energy when $Q$ is uniaxial. Of course, what may differ from one choice to another are the constraints on the coefficients under which the particular version is well-posed. On this point, it should be noted that one should expect {\it any} such choice of elastic energy density to impose restrictions more stringent than those dictated by the Ericksen inequalities on the $K_i$'s, cf. \eqref{ericksen} since one is seeking well-posedness over a much broader class of competitors than just uniaxials.} 

{\section{Appendix}
Here we elucidate the relationship between the energy \eqref{sigma} and the independent elastic invariants in \cite{LoMoTr87}. First, we observe that
\begin{equation}
\label{eq:divL}
\left|\left(\frac{s_0}{3}\id \pm Q\right)\dive Q\right|^2=\frac{s_0^2}{9}{\left|\dive Q\right|}^2\pm\frac{2s_0}{3}\dive Q\cdot Q\dive Q+{\left|Q\dive Q\right|}^2
\end{equation}
and, utilizing the anticommutator $\{A,B\}:=AB + BA$,
\begin{multline}
\label{eq:curlL}
\left|\left(\frac{s_0}{3}\id \pm Q\right)\curl Q\right|^2=\frac{1}{2}\sum_{j=1}^3\left|\left\{\frac{s_0}{6}\id \mp Q,\nabla Q_j-\nabla Q_j^T\right\}\right|^2\\=\frac{s_0^2}{9}\sum_{j=1}^3\left({\left|\nabla Q_j\right|}^2-\nabla Q_j\cdot\nabla Q_j^T\right)\qquad\qquad\\ \mp\frac{s_0}{3}\sum_{j=1}^3\left(\nabla Q_j-\nabla Q_j^T\right)\cdot\left(Q\left(\nabla Q_j-\nabla Q_j^T\right)+\left(\nabla Q_j-\nabla Q_j^T\right)Q\right)\\+{\left|\left(Q\left(\nabla Q_j-\nabla Q_j^T\right)+\left(\nabla Q_j-\nabla Q_j^T\right)Q\right)\right|}^2.
\end{multline}
The terms in \eqref{eq:divL} can be written as
\begin{align*}
    {\left|\dive Q\right|}^2&=Q_{ij,j}Q_{ik,k},\\
    \dive Q\cdot Q\dive Q&=Q_{im}Q_{ij,j}Q_{ml,l},\\
    {\left|Q\dive Q\right|}^2&=Q_{ik}Q_{il}Q_{kj,j}Q_{lm,m}.
\end{align*}
These correspond to the invariants $[L_2^{(2)}]$, $[L_3^{(3)}]$, and $[L^{(4)}_6]$ in \cite{LoMoTr87}, respectively. 
Expanding the terms quadratic in $Q$ in \eqref{eq:curlL} results in expressions
\begin{align*}
    {\left|\nabla Q_j\right|}^2&=Q_{ij,k}Q_{ij,k},\\
    \nabla Q_j\cdot \nabla Q_j^T&=Q_{ij,k}Q_{kj,i},
\end{align*}
corresponding to $[L_1^{(2)}]$ and $[L_3^{(2)}]$ in \cite{LoMoTr87}. The cubic terms in \eqref{eq:curlL} are
\begin{align*}
    \nabla Q_j\cdot Q\nabla Q_j&=Q_{il}Q_{ij,m}Q_{lj,m},\\
    \nabla Q_j\cdot Q\nabla Q_j^T&=Q_{il}Q_{ij,m}Q_{mj,l},\\
    \nabla Q_j\cdot \nabla Q_j\,Q&=Q_{il}Q_{mj,i}Q_{mj,l},
\end{align*}
and these are respectively the invariants $[L_4^{(3)}]$, $[L_6^{(3)}]$, and $[L^{(3)}_7]$ in \cite{LoMoTr87}. Note that $[L^{(3)}_7]$ is not an independent invariant (cf. \cite{LoMoTr87}) because it can be written as
\[[L^{(3)}_7]=2[L^{(3)}_6]+[L^{(3)}_5]-2[L^{(3)}_4]+[L^{(3)}_3]+2[L^{(3)}_2]-2[L^{(3)}_1].\]

Finally, the fourth order terms in \eqref{eq:curlL} are as follows
\begin{align*}
    Q\nabla Q_j\cdot Q\nabla Q_j&=Q_{il}Q_{im}Q_{lj,p}Q_{mj,p},\\
    Q\nabla Q_j\cdot Q\nabla Q_j^T&=Q_{il}Q_{lm}Q_{ik,p}Q_{pk,m},\\
    Q\nabla Q_j^T\cdot Q\nabla Q_j^T&=Q_{il}Q_{lm}Q_{kj,i}Q_{kj,m},\\
    Q\nabla Q_j\cdot \nabla Q_j\,Q&=Q_{ik}Q_{lm}Q_{lj,i}Q_{mj,k},\\
    Q\nabla Q_j\cdot \nabla Q_j^T\,Q&=Q_{ik}Q_{ml}Q_{lj,i}Q_{kj,m}.
\end{align*}
The first term in this list is the invariant $[L^{(4)}_7]$, the second term is $[L^{(4)}_9]$, and the fourth term is $[L^{(4)}_{11}]$. The third and the fifth terms do not belong to the list of independent invariants in \cite{LoMoTr87} and thus must be linear combinations of the invariants $[L^{(4)}_1]$ through $[L^{(4)}_{13}]$ in \cite{LoMoTr87}, cf. the following lemma. We will not pursue the issue of finding the coefficients of these combinations further.
\begin{lemma}\label{term35inv}
The terms 
\begin{equation}\label{term3}
Q\nabla Q_j^T \cdot Q \nabla Q_j^T = Q_{il}Q_{lm}Q_{kj,i}Q_{kj,m}
\end{equation}
and
\begin{equation}\label{term5}
    Q\nabla Q_j\cdot \nabla Q_j^T\,Q=Q_{ik}Q_{ml}Q_{lj,i}Q_{kj,m}
\end{equation}
are frame indifferent and materially symmetric, and so can each be written as a linear combination of the thirteen quartic $[L^{(4)}_j]$ invariants in \cite{LoMoTr87}.
\end{lemma}
\begin{proof}
We begin with the first term \eqref{term3} and follow the discussion in \cite{balltalk}. Both conditions involve invariance of \eqref{term3} under different coordinate changes $x \to z$, with corresponding measurements of the tensor $Q$, $Q^*$ in the $x=(x_1,x_2,x_3)$ (standard Cartesian) and $z=(z_1,z_2,z_3)$ (Cartesian after a change of variables) coordinates, respectively. More precisely, the frame indifference condition entails that for $\overline{x} \in \Omega$, $\tilde{R} \in SO(3)$ and $z=\overline{x}+\tilde{R}(x-\overline{x}),$
\begin{equation}\label{indiff}
    Q \nabla_x^T Q \cdot Q \nabla_x^T Q = Q^* \nabla_z^T Q^* \cdot  Q^* \nabla_z^T Q^*.
\end{equation}
Material symmetry is similar, except that $z=\overline{x}+\hat{R}(x - \overline{x})$, where $\hat{R}$ is the reflection $\hat{R} = -\id + 2 e \otimes e ,$ $|e|=1$. Since any $R \in O(3)$ can be expressed as $\hat{R}\tilde{R}$ for a reflection $\hat{R}$ and $\tilde{R} \in SO(3)$, verifying these two conditions can be combined and simplified into a single calculation, cf. \cite{balltalk}. Therefore, we must check that \eqref{indiff} holds for $R \in O(3)$. 
We will start on the right hand side of \eqref{indiff} and simplify, using the relation $R_{\gamma\alpha}R_{\gamma\beta} = \delta_{\alpha\beta}$ and the identities
\begin{align*}
    Q^*_{ab} &= R_{af}Q_{fg}R_{bg},\quad Q^*_{de,a} = R_{do}R_{ep}R_{aq}Q_{op,q},\\
    Q^*_{bc} &= R_{bh}Q_{hn}R_{cn},\quad
    Q^*_{de,c} = R_{dr}R_{es}R_{ct}Q_{rs,t},
\end{align*}
which can be derived as in \cite{balltalk}. Here $\delta_{\alpha\beta}$ is the usual Kronecker delta, and the first relation is a restatement of the identity $R^TR= \id$ for an orthogonal matrix. We write
\begin{align*}
    Q^* \nabla_z^T Q^* &\cdot  Q^* \nabla_z^T Q^*  \\
    &= Q^*_{ab}Q^*_{bc}Q^*_{de,a}Q^*_{de,c} \\
    &= R_{af}Q_{fg}R_{bg}R_{do}R_{ep}R_{aq}Q_{op,q}R_{bh}Q_{hn}R_{cn}R_{dr}R_{es}R_{ct}Q_{rs,t}\\
    &= R_{af}R_{aq}R_{bg}R_{bh}R_{do}R_{dr}R_{ep}R_{es}R_{ct}R_{cn}Q_{fg}Q_{hn}Q_{op,q}Q_{rs,t} \\
    &= \delta_{fq}\delta_{gh}\delta_{or}\delta_{ps}\delta_{tn}Q_{fg}Q_{hn}Q_{op,q}Q_{rs,t} \\
    &= Q_{fg}Q_{gt}Q_{op,f}Q_{op,t}.
\end{align*}
Of course the indices $f,g,t,o,p$ can be replaced by $i,j,k,l,m$, and we have shown that $Q^* \nabla_z^T Q^* \cdot Q^* \nabla_z^T Q^* = Q \nabla_x^T Q \cdot Q \nabla_x^T Q$.\par
Moving on to the second term, we first record the identities
\begin{align*}
    Q^*_{ab} &= R_{af}Q_{fg}R_{bg},\quad Q^*_{de,a} = R_{do}R_{ep}R_{aq}Q_{op,q},\\
    Q^*_{cd} &= R_{ch}Q_{hn}R_{dn},\quad
    Q^*_{be,c} = R_{br}R_{es}R_{ct}Q_{rs,t}.
\end{align*}
We may now calculate
\begin{align*}
    Q^*\nabla_z Q_j^* &\cdot \nabla_z^T Q_j^*\,Q^*  \\
    &= Q^*_{ab}Q^*_{cd}Q^*_{de,a}Q^*_{be,c}\\
    &= R_{af}Q_{fg}R_{bg}R_{ch}Q_{hn}R_{dn}R_{do}R_{ep}R_{aq}Q_{op,q}R_{br}R_{es}R_{ct}Q_{rs,t}\\
    &= R_{af}R_{aq}R_{bg}R_{br}R_{ch}R_{ct}R_{dn}R_{do}R_{ep}R_{es}Q_{fg}Q_{hn}Q_{op,q}Q_{rs,t} \\
    &= \delta_{fq}\delta_{gr}\delta_{ht}\delta_{no}\delta_{ps}Q_{fg}Q_{hn}Q_{op,q}Q_{rs,t}\\
    &=Q_{fg}Q_{hn}Q_{np,f}Q_{gp,h} \\
    &=  Q\nabla Q_j\cdot \nabla Q_j^T\,Q.
\end{align*}
\end{proof}}
\bibliographystyle{acm}
\bibliography{bibfiledissertation}

\end{document}